\documentclass[a4paper]{amsart}
\usepackage{amssymb,mathrsfs,mathtools}
\usepackage{setspace,enumerate}
\usepackage{bm}
\usepackage[utf8]{inputenc}
\usepackage[T1]{fontenc}

\newcommand{\N}{\mathbb{N}}
\newcommand{\C}{\mathbb{C}}
\newcommand{\R}{\mathbb{R}}

\newcommand{\lsz}{\left\lbrace}
\newcommand{\psz}{\right\rbrace}

\DeclareMathOperator{\sgn}{\mathrm{sgn}}
\newcommand{\dd}{\,\mathrm{d}}
\newcommand{\id}{Id}

\newtheorem{theorem}{Theorem}[section]
\newtheorem{lemma}[theorem]{Lemma}

\newtheorem{proposition}[theorem]{Proposition}
\newtheorem{remark}[theorem]{Remark}

\title{Atoms confined by very thin layers}
\author{\ Mat\v{e}j Tu\v{s}ek}
\date{28 June 2014}

\address{
Department of Mathematics,
Faculty of Nuclear Sciences and Physical Engineering,
Czech Technical University in Prague,
Trojanova 13, 120\,00 Prague 2, Czech Republic
}
\email{tusekmat@fjfi.cvut.cz}

\begin{document}

\begin{abstract}
The Hamiltonian of an atom with $N$ electrons and  a fixed nucleus of infinite mass between two parallel planes is considered in the limit when the distance $a$ between the planes tends to zero. We show that this Hamiltonian converges in the norm resolvent sense to a Schr\"{o}dinger operator acting effectively in $L^{2}(\R^{2N})$ whose potential part depends on $a$. Moreover, we prove that after an appropriate regularization this Schr\"{o}dinger operator tends, again in the norm resolvent sense, to the Hamiltonian of a two-dimensional atom (with the three-dimensional Coulomb potential-one over distance), as $a\to 0$. This makes possible to locate the discrete spectrum of the full Hamiltonian once we know the spectrum of the latter one. Our results also provide a mathematical justification for the interest in the two-dimensional atoms with the three-dimensional Coulomb potential.
\end{abstract}

\maketitle

\section{Introduction}
In this paper we discuss a non-relativistic quantum system of an atom confined in the middle of an infinite planar layer with impenetrable walls. We describe it by the three-dimensional atomic Hamiltonian with  the Dirichlet boundary condition imposed on the boundary planes. The study of confined atomic systems has a long history (\cite{MiBoBi_37}, from 1937, and \cite{GrSe_46}, from 1946, deal with the hydrogen atom with a nucleus placed at the center of an impenetrable spherical box of finite radius), as these systems may serve as important models for caged and compressed atoms \cite{FrYnFr_87,Ja_96,LaBuCo_02,Ki_09} or hydrogenic impurities in quantum dots \cite{Va_99,LiAmDo_08}. In the above mentioned references, only the confinement to finite regions, usually to balls, is considered. However, with prospects of mesoscopic physics applications and also for richer mathematical properties of corresponding Hamiltonians (presence of the continuous spectrum), a hydrogen atom confined in regions that are 
unbounded in some directions has quite recently drawn interest.

In particular, \cite{DuHo_10} deals with a hydrogen atom confined by a straight infinite tube, whereas in \cite{DuStTu_10} a hydrogen atom in a thin infinite planar layer of width $a$ was studied. The present paper may be viewed as an extension of the results obtained in the latter source to a multi-electron case. Therein the so-called effective Hamiltonian was introduced as a projection of the full Hamiltonian to the lowest transverse mode of the Dirichlet Laplacian on the layer. Due to the large separation distance between subsequent eigenvalues of the Dirichlet Laplacian in the transverse direction (it is proportional to $a^{-2}$), it was demonstrated that the effective Hamiltonian well approximates  the full atomic Hamiltonian in the norm resolvent sense, as $a\to 0$. After an appropriate regularization, the effective Hamiltonian may be in turn approximated by the Hamiltonian of a two-dimensional hydrogen atom (with the three-dimensional Coulomb potential, i.e., one over distance), as $a\to 0$. Since the spectrum of 
the latter Hamiltonian is explicitly known, one can use it to approximate the spectrum of the full Hamiltonian. 

Let us stress that there are several new aspects that complicates a similar analysis in the multi-electron case. First of all, the repulsive electron to electron interaction is involved. With some effort, we will be able to control it in a similar manner as the electron to nucleus interaction in the single electron case. Next, we must take the fermionic nature of electrons into the account. Actually, we will treat electrons  as distinguishable particles and only at the very end of the paper we perform reduction to the subspace of totally antisymmetric functions. Finally, let us recall that the spectrum of a two-dimensional (as well as a three-dimensional) atom is not known explicitly except of the single electron case.  Nevertheless, there are still some qualitative spectral results for two-dimensional atoms. See \cite{NaPoSo_11} for a concise presentation of them. At this point, let us make clear that in the present paper we do not concern with the question of the maximal negative ionization of 
the confined or two-dimensional atom (the existence of bound state). For our purposes, the stability of the first type (the lower boundness), which is quite easy to prove, will be sufficient. 

The paper is organized as follows. In Sections \ref{sec:2D}, \ref{sec:layer}, and \ref{sec:eff}, the Hamiltonians of a two-dimensional atom and an atom in a planar layer, and the effective Hamiltonian, respectively, are introduced in detail as self-adjoint operators. In the next sections, relations between these Hamiltonians are given. The main theorem comes in Section \ref{sec:main}. It essentially claims that the full Hamiltonian tends, in the norm resolvent sense, to the Hamiltonian of a two-dimensional atom, as $a\to 0$. (See Theorem \ref{theo:main} for a precise formulation.) Therefore the two-dimensional atom, which is a kind of mathematical construction, may be viewed as a limit of a physical system of an atom compressed among a pair of parallel planes. In this context, let us remark that the two-dimensional hydrogen atom (with the three-dimensional Coulomb potential) is of continuous interest in the literature \cite{PaPo_02,RoRo_03,ChaLe_05,Gr_08}.
Section \ref{sec:spectrum} is devoted to the localization of the discrete spectrum of the full Hamiltonian. Also analyticity of its eigenvalues in $a$ is briefly studied. Finally, in Section \ref{sec:fermionized}, we discuss the fermionized versions of the Hamiltonians and conclude that the approximation results remain valid.

\section{Hamiltonian of a two-dimensional atom}\label{sec:2D}
Consider $N$ mutually interacting electrons with the unit charge and mass in the field of a nucleus with an atomic charge $Z>0$ and infinite mass. Denote by $\bm{\varrho}_{i}\equiv(x_{i},y_{i})$ the coordinate of the $i$th electron in the center of mass coordinate system and introduce the following notation
$$\varrho_{i}:=|\bm{\varrho}_{i}|,\quad \varrho_{i,j}:=|\bm{\varrho}_{i}-\bm{\varrho}_{j}|.$$   
Then the Hamiltonian of this system, $h_{Z,N}$, is given by
\begin{align*}
 &h_{N,Z}:=-\Delta_{\R^{2N}}-\sum_{i=1}^{N}\frac{Z}{\varrho_{i}}+\sum_{1\leq i<j\leq N}\frac{1}{\varrho_{i,j}}\\
 &Q(h_{N,Z}):=Q(-\Delta_{\R^{2N}})=\mathcal{H}^{1}(\R^{2})^{\otimes^{N}}\equiv \mathcal{H}^{1}(\R^{2N}).
\end{align*}
Here $\Delta_{\R^{2N}}=\sum_{i=1}^{N}\Delta_{\bm{\varrho}_i}$, where $\Delta_{\bm{\varrho}_i}$ stands for the Laplacian in the $i$th coordinate (naturally extended on the appropriate tensor product), and $Q$ denotes the form domain of an operator. The operator $-\Delta_{\R^{2N}}$ with the chosen form domain is self-adjoint \cite{rs2} and so is $h_{N,Z}$ by the KLMN theorem, as we will prove below.

\begin{lemma}\label{lem:int_term_bound}
 For any $\psi\in \mathcal{H}^{1}(\R^{4})$, it holds
$$\langle\psi,\,\varrho_{i,j}^{-1}\psi\rangle\leq \frac{\Gamma(1/4)^{4}}{4\pi^{2}\sqrt{2}}\langle\psi,\,\sqrt{-\Delta_{\bm{\varrho}_i}-\Delta_{\bm{\varrho}_j}}\psi\rangle.$$
\end{lemma}
\begin{proof}
 With the aid of a unitary mapping $U:L^{2}(\R^{4})\to L^{2}(\R^{4})$,
$$(U\psi)(\bm{s},\bm{t})=\psi\left(\frac{\bm{s}-\bm{t}}{\sqrt{2}},\frac{\bm{s}+\bm{t}}{\sqrt{2}}\right),$$
 the Fubini theorem, and the two-dimensional Kato inequality (see \cite{He_77}, \cite{Bo_02}),
 $$\frac{1}{\varrho}\leq\frac{\Gamma(1/4)^{4}}{4\pi^{2}}\sqrt{-\Delta_{\bm{\varrho}}},\quad\text{on }\mathcal{H}^{1}(\R^{2}),$$
 we have
\begin{equation}\label{eq:int_term_bound}
 \begin{split}
 \langle\psi,\,\varrho_{i,j}^{-1}\psi\rangle&=\langle U\psi,\, U\varrho_{i,j}^{-1}U^{\dagger}U\psi\rangle=\langle U\psi,\, (\sqrt{2}|\bm{t}|)^{-1}U\psi\rangle\\
 &\leq\frac{\Gamma(1/4)^{4}}{4\pi^{2}\sqrt{2}}\langle U\psi,\, \sqrt{-\Delta_{\bm{t}}}\,U\psi\rangle.
 \end{split}
\end{equation}
Passing to the Fourier image in variables $\bm{t}$ and $\bm{s}$, we directly infer that
\begin{equation}\label{eq:laplace_est}
\sqrt{-\Delta_{\bm{t}}}\leq \sqrt{-\Delta_{\bm{t}}-\Delta_{\bm{s}}}=U\sqrt{-\Delta_{\bm{\varrho}_i}-\Delta_{\bm{\varrho}_j}}\,U^{\dagger}.
\end{equation}
Putting this into \eqref{eq:int_term_bound}, we obtain the assertion of the lemma.
\end{proof}

\begin{proposition}\label{prop:klmn_2d}
 Denote 
$$\mathcal{V}_{\mathrm{2D}}:=-\sum_{i=1}^{N}\frac{Z}{\varrho_{i}}+\sum_{1\leq i<j\leq N}\frac{1}{\varrho_{i,j}}.$$
Then for any $\epsilon>0$ and $\psi\in H^{1}(\R^{2N})$,
\begin{equation*}
 |\langle\psi,\mathcal{V}_{\mathrm{2D}}\psi\rangle|\leq\frac{\Gamma(1/4)^{4}}{8\pi^{2}}\sqrt{N}\mathrm{max}\lsz\frac{N-1}{\sqrt{2}},Z\psz \left(\epsilon\|\nabla\psi\|^{2}+\epsilon^{-1}\|\psi\|^{2}\right),
\end{equation*}
where $\nabla$ is the $2N$-dimensional gradient,
and
\begin{equation}\label{eq:num_val_bound}
h_{N,Z}\geq-N\left(\frac{\Gamma(1/4)^{4}}{8\pi^{2}}Z\right)^{2}. 
\end{equation}
\end{proposition}
\begin{proof}
 Let $\hat{\psi}$ (in a variable $\bm{\lambda}\equiv(\bm{\lambda}_{1},\ldots,\bm{\lambda}_{N})$) stands for the Fourier image of $\psi$. Then by the two-dimensional Kato inequality,
\begin{equation}\label{eq:2d_pot_neg}
\begin{split}
 -\langle\psi,\mathcal{V}_{\mathrm{2D}}\psi\rangle&\leq\langle\psi,Z \sum_{i}\varrho_{i}^{-1}\psi\rangle\leq\frac{\Gamma(1/4)^{4}}{4\pi^{2}}Z\langle\hat{\psi},\sum_{i}|\bm{\lambda}_{i}|\hat{\psi}\rangle\leq\frac{\Gamma(1/4)^{4}}{4\pi^{2}}Z\langle\hat{\psi},\sqrt{N}|\bm{\lambda}|\hat{\psi}\rangle\\
 &\leq \frac{\Gamma(1/4)^{4}}{4\pi^{2}}Z\sqrt{N}\|\psi\|\|\nabla\psi\|\leq\frac{\Gamma(1/4)^{4}}{8\pi^{2}}Z\sqrt{N}\left(\epsilon\|\nabla\psi\|^{2}+\epsilon^{-1}\|\psi\|^{2}\right).
\end{split}
\end{equation}
Using Lemma \ref{lem:int_term_bound} we obtain in a similar manner as above,
\begin{equation*}
\begin{split}
 &\langle\psi,\mathcal{V}_{\mathrm{2D}}\psi\rangle\leq\langle\psi,\sum_{i<j}\varrho_{i,j}^{-1}\psi\rangle\leq \frac{\Gamma(1/4)^{4}}{4\pi^{2}\sqrt{2}}\langle\hat{\psi},\sum_{i<j}\sqrt{\bm{\lambda}_{i}^{2}+\bm{\lambda}_{j}^{2}}~\hat{\psi}\rangle\\
 &\leq \frac{\Gamma(1/4)^{4}}{4\pi^{2}\sqrt{2}}(N-1)\langle\hat{\psi},\sum_{i}|\bm{\lambda}_{i}|\hat{\psi}\rangle
 \leq \frac{\Gamma(1/4)^{4}}{8\pi^{2}\sqrt{2}}(N-1)\sqrt{N}\left(\epsilon\|\nabla\psi\|^{2}+\epsilon^{-1}\|\psi\|^{2}\right).
\end{split}
\end{equation*}
Neglecting the positive component of $\mathcal{V}_{\mathrm{2D}}$, \eqref{eq:num_val_bound} follows from \eqref{eq:2d_pot_neg} with
$$\varepsilon=\left(\frac{\Gamma(1/4)^{4}}{8\pi^{2}}Z\sqrt{N}\right)^{-1}.$$
\end{proof}

\begin{remark}[Spectrum of $h_{N,Z}$]\label{rem:h_spec}
 The spectrum of $h_{N,Z}$ is, of course, explicitly known only if $N=1$. In that case \cite{YaGuCha_91,DuStTu_10},
 $$\sigma_{p}(h_{1,Z})=\left\{-\frac{Z^{2}}{(2N-1)^{2}},\, N\in\N\right\},\quad \sigma_{ess}(h_{1,Z})=\sigma_{ac}(h_{1,Z})=[0,\infty).$$
 For $N>1$, we have the HVZ theorem \cite{Si_77} which states that
 $$\sigma_{ess}(h_{N,Z})=[\inf\sigma(h_{N-1,Z}),\infty).$$
\end{remark}

\section{Hamiltonian of an atom in a layer}\label{sec:layer}
Let $\Omega_{a}=\R^{2}\times (-a/2,a/2)$ with $a>0$. Consider a three-dimensional atom with $N$ electrons and and with a nucleus of  infinite mass and of  a charge $Z>0$ restricted to $\Omega_{a}$ by imposing the Dirichlet boundary condition on the boundary planes. For simplicity, let us place the nucleus at the origin. Then the Hamiltonian, $H_{N,Z}^{a}$, of this system acts in $L^{2}(\Omega_{a})^{\otimes^{N}}$ as follows,
\begin{equation*}
H_{N,Z}^{a}:=-\Delta_{\Omega_{a}^{N}}-\sum_{i=1}^{N}\frac{Z}{r_{i}}+\sum_{1\leq i<j\leq N}\frac{1}{r_{i,j}},
\end{equation*}
where $\bm{r}_{i}\equiv(x_{i},y_{i},z_{i})\in\Omega_{a}$ is the coordinate of the $i$th electron,
$$r_{i}:=|\bm{r}_{i}|,\quad r_{i,j}:=|\bm{r}_{i}-\bm{r}_{j}|,$$
and $-\Delta_{\Omega_{a}^{N}}$ stands for the free Hamiltonian of an $N$-particle system in $\Omega_a$. In more detail,
\begin{align*}
&-\Delta_{\Omega_{a}^{N}}:=-\sum_{i=1}^{N}\id\otimes\ldots\otimes\Delta_{\bm{r}_{i}}\otimes\ldots\otimes\id\\
&\mathrm{Dom}(-\Delta_{\Omega_{a}^{N}}):=\big(\mathcal{H}_{0}^{1}(\Omega_{a})\cap\mathcal{H}^{2}(\Omega_{a})\big)^{\otimes^{N}},
\end{align*}
where $\Delta_{\bm{r}_{i}}$ is the Laplace operator on $L^{2}(\Omega_{a})$ (in the variable $\bm{r}_{i}$) with the Dirichlet boundary condition. This operator is self-adjoint \cite{davies}. Below we will show that $H_{N,Z}^{a}$ is $(-\Delta_{\Omega_{a}^{N}})$-bounded with a relative bound smaller than one. Thus $H_{N,Z}^{a}$ is also self-adjoint on
$\mathrm{Dom}(H_{N,Z}^{a})=\mathrm{Dom}(-\Delta_{\Omega_{a}^{N}})$ by the Kato-Rellich theorem.

Put
$$\mathcal{V}:=-\sum_{i=1}^{N}\frac{Z}{r_{i}}+\sum_{1\leq i<j\leq N}\frac{1}{r_{i,j}}.$$
Then
$$\|\mathcal{V}\psi\|^{2}\leq Z^{2}N\sum_{i=1}^{N}\|r_{i}^{-1}\psi\|^{2}+\binom{N}{2}\sum_{1\leq i<j\leq N}\|r_{i,j}^{-1}\psi\|^{2}.$$
Take $\psi\in \big(\mathcal{H}_{0}^{1}(\Omega_{a})\cap\mathcal{H}^{2}(\Omega_{a})\big)^{\otimes N}$.
Recall that the Hardy inequality (see, e.g., \cite{WaZh_03}) states 
\begin{equation}\label{eq:Hardy_ineq}
\frac{1}{4}\int_{\mathbb{R}^{3}}\frac{|u(\mathbf{x})|^{2}}{|\mathbf{x}|^{2}}\,\mathrm{d}\mathbf{x}\leq\int_{\mathbb{R}^{3}}|\nabla u(\mathbf{x})|^{2}\,\mathrm{d}\mathbf{x}
\end{equation}
for any  $u\in\mathcal{D}^{1,2}(\mathbb{R}^{3})$, and that $\mathcal{H}_{0}^{1}(\Omega_{a})$ may be naturally embedded into $\mathcal{D}^{1,2}(\R^{3})$.
Thus, using the Fubini theorem, we have
\begin{equation*}
 \|r_{i}^{-1}\psi\|^{2}\leq 4\|\nabla_{\bm{r}_{i}}\psi\|^{2}=4\langle\psi,-\Delta_{\bm{r}_{i}}\psi\rangle.
\end{equation*}
Similarly, we obtain
\begin{align*}
 \|r_{1,2}^{-1}&\psi\|^{2}=\int_{\Omega_{a}^{N}}\frac{1}{|\bm{r}_{1}-\bm{r}_{2}|^{2}}|\psi(\bm{r}_{1},\ldots\bm{r}_{N})|^{2}\dd\bm{r}_{1}\ldots\dd\bm{r}_{N}\\
 &=\int_{\Omega_{a}^{\times (N-1)}}\int_{\R^{2}\times(-a/2-z_2,a/2-z_2)}\frac{1}{|\bm{u}|^{2}}\,|\psi(\bm{u}+\bm{r}_2,\bm{r}_2,\ldots\bm{r}_{N})|^{2}\dd\bm{u}\dd\bm{r}_2\ldots\dd\bm{r}_{N}\\
 &\leq 4\int_{\Omega_{a}^{\times (N-1)}}\int_{\R^{2}\times(-a/2-z_2,a/2-z_2)}|\nabla_{\bm{u}}\psi(\bm{u}+\bm{r}_2,\bm{r}_2,\ldots\bm{r}_{N})|^{2}\dd\bm{u}\dd\bm{r}_{2}\ldots\dd\bm{r}_{N}\\
 &=4\int_{\Omega_{a}^{N}}|\nabla_{\bm{r}_{1}}\psi(\bm{r}_{1},\ldots\bm{r}_{N})|^{2}\dd\bm{r}_{1}\ldots\dd\bm{r}_{N}= 4\langle\psi,-\Delta_{\bm{r}_{1}}\psi\rangle.
\end{align*}
We conclude that
\begin{equation}\label{eq:full_pot_bound}
\begin{split}
 \|\mathcal{V}\psi\|^{2}&\leq 4\left(Z^{2}N+\binom{N}{2} (N-1)\right)\langle\psi,-\Delta_{\Omega_{a}^{N}}\psi\rangle\\
 &\leq \left(2Z^{2}N+N(N-1)^{2}\right)\Big(\epsilon\|-\Delta_{\Omega_{a}^{N}}\psi\|^{2}+\epsilon^{-1}\|\psi\|^{2}\Big)
\end{split}
\end{equation}
for any $\epsilon>0$.

\section{The effective Hamiltonian}\label{sec:eff}
Let us decompose the kinetic part of $H_{N,Z}^{a}$ with respect to the transverse modes,
$$-\Delta_{\Omega_{a}^{N}}=\sum_{i=1}^{N}\left[\bigoplus_{n=1}^{\infty}(-\Delta_{\bm{\varrho}_i}+E_{n}^{a})\otimes\langle .,\chi_{n}^{a}\rangle\chi_{n}^{a}\right],$$
where 
\begin{equation*}
  E_{n}^{a}:=(n\pi/a)^{2},\quad\chi^{a}_{n}(z):=\sqrt{\frac{2}{a}}\begin{cases}
                               \cos\frac{n\pi z}{a}& \text{ if n is odd}\\
				\sin\frac{n\pi z}{a}& \text{ if n is even.}
                              \end{cases}
\end{equation*}
The {\it effective Hamiltonian}, $H_{\mathrm{eff}}^{a}$, is defined with the aid of the projection on the lowest transverse mode,
$$P^{a}=\bigotimes_{i=1}^{N}\left[\id_{L^{2}(\R^{2})}\otimes\langle .,\chi_{1}^{a}(z_{i})\rangle\chi_{1}^{a}(z_{i})\right],$$
as 
$$H_{\mathrm{eff}}^{a}:=P^{a}H_{N,Z}^{a}P^{a}.$$
It is well defined  on $\mathrm{Dom}(H_{N,Z}^{a})$ because $\mathrm{Dom}(H_{N,Z}^{a})$ is invariant under $P^{a}$. Moreover,
$P^{a}\mathrm{Dom}(H_{N,Z}^{a})$ is naturally isometric to $\mathcal{H}^{2}(\R^{2})^{\otimes^{N}}$. Thus we will view $H_{\mathrm{eff}}^{a}$ as the following operator on $L^{2}(\R^{2N})$,
\begin{align*}
&H_{\mathrm{eff}}^{a}=-\Delta_{\R^{2N}}+NE_{1}^{a}-Z\sum_{i=1}^{N}V_{en}^{a}(\varrho_{i})+\sum_{1\leq i<j\leq N}V_{ee}^{a}(\varrho_{i,j})\\
&\mathrm{Dom}(H_{\mathrm{eff}}^{a})=\mathcal{H}^{2}(\R^{2})^{\otimes^{N}}.
\end{align*}
Here the {\it effective potentials} are defined by
\begin{align*}
 &V_{en}^{a}(\varrho):=\frac{4}{a}\int_{0}^{a/2}\frac{\cos^2\frac{\pi s}{a}}{\sqrt{\varrho^2+s^2}}\dd s,\\
 &V_{ee}^{a}(\varrho):=\frac{4}{a^{2}}\int_{-a/2}^{a/2}\int_{-a/2}^{a/2}\frac{\cos^{2}\frac{\pi s}{a}\cos^{2}\frac{\pi t}{a}}{\sqrt{\varrho^{2}+(s-t)^{2}}}\dd s\dd t.
\end{align*}
Below we will prove self-adjointness of $H^{a}_{\mathrm{eff}}$.

At first, we need to know some properties of the effective potentials.
The properties of $V_{en}^{a}$ are extensively discussed in \cite{DuStTu_10} (therein $V_{en}^{a}$ is called simply $V_{\mathrm{eff}}^{a}$), but we will also summarize some of them here for the reader's convenience. One directly verifies that
\begin{description}
 \item[Asymptotic behavior]
 \begin{align*}
 &V_{i}^{a}(\varrho)=\frac{1}{\varrho}+O\left(\frac{1}{\varrho^{3}}\right),\quad\text{as }\varrho\to\infty,\quad i\in\{ee,en\}\\
 &V_{ee}^{a}(\varrho)=-\frac{3}{a}\ln{\varrho}+O(1)\text{ and }V_{en}^{a}(\varrho)=-\frac{4}{a}\ln{\varrho}+O(1),\quad\text{as }\varrho\to 0+
 \end{align*}
 \item[Scaling properties]
 $$V_{i}^{a}(\varrho)=\frac{1}{a}V_{i}^{1}\left(\frac{\varrho}{a}\right),\quad i\in\{ee,en\}$$
 \item[Bounds]
 \begin{equation}\label{eq:pot_bound}
 0\leq V_{i}^{a}(\varrho)\leq\frac{1}{\varrho},\quad i\in\{ee,en\}.
 \end{equation}
\end{description}
Moreover $V_{ee}^{a}$ and $V_{en}^{a}$  are strictly decreasing.

Now we see that 
$$V_{i}^{a}\in L^{2}(\R^{2})+L^{\infty}(\R^{2}),\quad i\in\{ee,en\}.$$
Hence, in the similar manner as in the three-dimensional case (see \cite[Thm. X.16]{rs2}), it follows that $H_{\mathrm{eff}}^{a}$ is self-adjoint on $\mathcal{H}^{2}(\R^{2})^{\otimes^{N}}$,
due to the Kato-Rellich theorem.  

Furthermore, one can easily observe that
$$ W_{i}(\varrho):=1-\varrho V_{i}^{1}(\varrho),\quad i\in\{ee,en\},$$
are $L^{1}(\R_{+},\dd\varrho)$-functions and  $0\leq W_{i}\leq 1$. Consequently, we can apply \cite[Lemma 5]{DuStTu_10}. Here we reproduce it in a slightly modified form.

\begin{lemma}\label{lem:gen_dif_est}
Suppose $W\in L^{1}(\mathbb{R}_{+},\mathrm{d}\varrho)$
and $0\leq W\leq1$. Put \begin{equation}
V^{a}(\varrho)=\frac{1}{\varrho}\left(1-W\!\left(\frac{\varrho}{a}\right)\right)\!,\ a>0.\label{eq:Va_def}\end{equation}
Then for any $a$, $0<a<1/2$, one has \begin{equation}
\begin{split} \big\|(-\Delta_{\R^2}+2)^{-1/2}&\left(\varrho^{-1}-V^{a}\right)(-\Delta_{\R^2}+2)^{-1/2}\big\|\\
\noalign{\smallskip} & \leq\,2\sqrt{3}~a|\ln{a}|\int_{\mathbb{R}_{+}}W(\varrho)\,\mathrm{d}\varrho+4\sqrt{2}~a\left(\int_{\mathbb{R}_{+}}W(\varrho)\,\mathrm{d}\varrho\right)^{1/2},\end{split}
\label{eq:dif_est}\end{equation}
where $-\Delta_{\R^2}$ stands for the two-dimensional free particle Hamiltonian.
\end{lemma}

\begin{remark}[Spectrum of $H_{\mathrm{eff}}^{a}$]\label{rem:eff_spectrum}
 Due to \eqref{eq:pot_bound}, $\sigma(H_{\mathrm{eff}}^{a}-NE_{1}^{a})$ has a lower bound given in \eqref{eq:num_val_bound}.
 Moreover, the HVZ theorem yields 
$[NE_{1}^{a},\infty)\,\subset\sigma_{ess}(H_{\mathrm{eff}}^{a}).$
\end{remark}

\section{Approximation of the effective Hamiltonian by the two-dimensional atomic Hamiltonian}
Observe that single-particle potentials that are controlled by $-\Delta_{\R^2}$ are also controlled by $-\Delta_{\R^{2N}}$, which will be henceforth denoted just by $-\Delta$ to lighten the notation. The same holds true for electron-to-electron interaction terms (that are only functions of the mutual distance), since they may be viewed as single-particle potentials with the appropriate change of coordinates. In detail, for $\xi>0$,
$$(-\Delta_{\R^2}+\xi)\otimes\id\leq -\Delta+\xi$$
which implies \cite[Thm. VI.2.21]{kato}
$$(-\Delta+\xi)^{-1}\leq(-\Delta_{\R^2}+\xi)^{-1}\otimes\id.$$
In particular, for all $\psi\in Q(V\otimes\id)$, where $V\geq 0$ and $Q(V)=Q(-\Delta_{\R^2})=\mathcal{H}^{1}(\R^{2})$,
$$\langle\psi,\sqrt{V\otimes\id}~(-\Delta+\xi)^{-1}\sqrt{V\otimes\id}~\psi\rangle\leq\big\langle\psi,\big[\sqrt{V}~(-\Delta_{\R^2}+\xi)^{-1}\sqrt{V}\big]\otimes\id~\psi\big\rangle$$
from which it follows that
$$\|\sqrt{V\otimes\id}~(-\Delta+\xi)^{-1}\sqrt{V\otimes\id}\|\leq\|\sqrt{V}~(-\Delta_{\R^2}+\xi)^{-1}\sqrt{V}\|$$
or equivalently
\begin{equation}\label{eq:free_control}
\|(-\Delta+\xi)^{-1/2}(V\otimes\id)(-\Delta+\xi)^{-1/2}\|\leq\|(-\Delta_{\R^2}+\xi)^{-1/2}V(-\Delta_{\R^2}+\xi)^{-1/2}\|. 
\end{equation}

Let us denote
\begin{equation*}
 \mathscr{V}^{a}:=Z\sum_{i=1}^{N}\left(\frac{1}{\varrho_{i}}-V_{en}^{a}(\varrho_{i})\right)-\sum_{1\leq i<j\leq N}\left(\frac{1}{\varrho_{i,j}}-V_{ee}^{a}(\varrho_{i,j})\right).
\end{equation*}

\begin{proposition}\label{prop:eff_2d_est}
For any $0<a<1/2$, we have
 \begin{equation}\label{eq:eff_2d_est}
 \begin{split}
 \|(-\Delta+2)^{-1/2}|\mathscr{V}^{a}|&(-\Delta+2)^{-1/2}\|\\
 \leq& 2\sqrt{3}\left[ZN\left(\frac{1}{4}-\frac{1}{\pi^2}\right)+\binom{N}{2}\left(\frac{1}{3}-\frac{5}{4\pi^2}\right)\right]a |\ln a|\\
 &+4\sqrt{2}\left[ZN\left(\frac{1}{4}-\frac{1}{\pi^2}\right)^{1/2}+\binom{N}{2}\left(\frac{1}{3}-\frac{5}{4\pi^2}\right)^{1/2}\right]a.
 \end{split}
 \end{equation}
\end{proposition}
\begin{proof}
By the triangle inequality and \eqref{eq:free_control},
 \begin{equation*}
 \begin{split}
  \|(-\Delta&+2)^{-1/2}|\mathscr{V}^{a}|(-\Delta+2)^{-1/2}\|\\
  \leq& ZN\big\|(-\Delta_{\bm{\varrho}_{1}}+2)^{-1/2}(\varrho_{1}^{-1}-V_{en}^{a}(\varrho_{1}))(-\Delta_{\bm{\varrho}_{1}}+2)^{-1/2}\big\|\\
  &+\binom{N}{2}\big\|(-\Delta_{\bm{\varrho}_{1}}-\Delta_{\bm{\varrho}_{2}}+2)^{-1/2}(\varrho_{1,2}^{-1}-V_{ee}^{a}(\varrho_{1,2}))(-\Delta_{\bm{\varrho}_{1}}-\Delta_{\bm{\varrho}_{2}}+2)^{-1/2}\big\|.
 \end{split}
 \end{equation*}
Put $\bm{t}:=2^{-1/2}(\bm{\varrho}_2-\bm{\varrho}_1)=:2^{-1/2}\bm{\varrho}_{1,2}$. Then in the second term we may estimate as follows,
\begin{align*}
 \big\|(-\Delta_{\bm{\varrho}_{1}}-\Delta_{\bm{\varrho}_{2}}+2)^{-1/2}&(\varrho_{1,2}^{-1}-V_{ee}^{a}(\varrho_{1,2}))(-\Delta_{\bm{\varrho}_{1}}-\Delta_{\bm{\varrho}_{2}}+2)^{-1/2}\big\|\\
 &\leq\big\|(-\Delta_{\bm{t}}+2)^{-1/2}(\varrho_{1,2}^{-1}-V_{ee}^{a}(\varrho_{1,2}))(-\Delta_{\bm{t}}+2)^{-1/2}\big\|.
\end{align*}
Here we employed \eqref{eq:laplace_est} and \eqref{eq:free_control}.
Since 
$$\|(-\Delta_{\bm{\varrho}_{1,2}}+2)^{1/2}(-\Delta_{\bm{t}}+2)^{-1/2}\|^{2}=\frac{1}{2}\|(-\Delta_{\bm{\varrho}_{1,2}}+2)^{1/2}(-\Delta_{\bm{\varrho}_{1,2}}+1)^{-1/2}\|^{2}=1,$$
we obtain 
\begin{align*}
\|(-\Delta_{\bm{\varrho}_{1}}-\Delta_{\bm{\varrho}_{2}}+2)^{-1/2}&(\varrho_{1,2}^{-1}-V_{ee}^{a}(\varrho_{1,2}))(-\Delta_{\bm{\varrho}_{1}}-\Delta_{\bm{\varrho}_{2}}+2)^{-1/2}\big\|\\
&\leq \|(-\Delta_{\bm{\varrho}_{1,2}}+2)^{-1/2}(\varrho_{1,2}^{-1}-V_{ee}^{a}(\varrho_{1,2}))(-\Delta_{\bm{\varrho}_{1,2}}+2)^{-1/2}\big\|.
\end{align*}
Lemma \ref{lem:gen_dif_est} now yields
\begin{equation*}
 \begin{split}
 \|(-\Delta+2)^{-1/2}&|\mathscr{V}^{a}|(-\Delta+2)^{-1/2}\|\\
 \leq& 2\sqrt{3}\left[ZN\left(\int_{\R_{+}}W_{en}(\varrho)\dd\varrho\right)+\binom{N}{2}\left(\int_{\R_{+}}W_{ee}(\varrho)\dd\varrho\right)\right]a |\ln a|\\
 &+4\sqrt{2}\left[ZN\left(\int_{\R_{+}}W_{en}(\varrho)\dd\varrho\right)^{1/2}+\binom{N}{2}\left(\int_{\R_{+}}W_{ee}(\varrho)\dd\varrho\right)^{1/2}\right]a.
 \end{split}
 \end{equation*}
The integrals of $W_{en}$ and $W_{ee}$ may be evaluated using Fubini's theorem,
$$\int_{\R_{+}}W_{en}(\varrho)\dd\varrho=\frac{1}{4}-\frac{1}{\pi^2},\quad \int_{\R_{+}}W_{ee}(\varrho)\dd\varrho=\frac{1}{3}-\frac{5}{4\pi^2},$$
which completes the proof.
\end{proof}

Further we will need an estimate formulated in the following auxiliary lemma that in fact is a standard result (see, e.g., \cite[Chpt.~XI]{rs3}).

\begin{lemma} \label{lem:symm_resolv_eq} Assume that $A$ is semi-bounded,
$A^{-1}$ exists and is bounded, $C$ is self-adjoint and $A$ form
bounded. If\[
\alpha=\||C|^{1/2}|A|^{-1/2}\|<1\]
then $(A+C)^{-1}$ exists, is bounded and\[
\|(A+C)^{-1}-A^{-1}\|\leq\frac{\alpha^{2}\|A^{-1}\|}{1-\alpha^{2}}\,.\]
\end{lemma}

\begin{theorem}\label{theo:eff_2d_coul}
Let $d_{N,Z}(\xi):=\mathrm{dist}(\xi,\sigma(h_{N,Z}))$ and $\mu\leq\inf\sigma(h_{N,Z})-2$. Then, for every $\xi\in\mathrm{Res}(h_{N,Z})\cap\R$, there exists $a_{0}(\xi)>0$ (which is given within the proof) such that for all $a$, $0<a<a_{0}(\xi)$, one has $\xi\in\mathrm{Res}(H_{\mathrm{eff}}^{a}-NE_{1}^{a})$ and
\begin{equation*}
\begin{split}
\|(H_{\mathrm{eff}}^{a}-NE_{1}^{a}-\xi)^{-1}&-(h_{N,Z}-\xi)^{-1}\|\\
&\leq\frac{2}{d_{N,Z}(\xi)}\max\lsz 1,\frac{-\mu}{d_{N,Z}(\xi)}\psz C_{1}(N,Z)^{2}C_{2}(N,Z)\,a|\ln{a}|,
\end{split}
\end{equation*}
The constants $C_{1}(N,Z)$ and $C_{2}(N,Z)$ are given by \eqref{eq:free_vs_coul_bound} and  \eqref{eq:pot_diff_bound}, respectively.
\end{theorem}
\begin{proof}
 In Lemma \ref{lem:symm_resolv_eq}, we set
$$A=h_{N,Z}-\xi,\quad C=\mathscr{V}^{a}.$$
Then $\|A^{-1}\|=d_{N,Z}(\xi)^{-1}$ and
\begin{align*}
 \alpha^{2}=&\||h_{N,Z}-\xi|^{-1/2}|\mathscr{V}^{a}||h_{N,Z}-\xi|^{-1/2}\|\leq \|(h_{N,Z}-\mu)^{1/2}|h_{N,Z}-\xi|^{-1/2}\|^{2}\\
 &\times \|(-\Delta-\mu)^{1/2}(h_{N,Z}-\mu)^{-1/2}\|^{2}\|(-\Delta+2)^{1/2}(-\Delta-\mu)^{-1/2}\|^{2}\\
 &\times \|(-\Delta+2)^{-1/2}|\mathscr{V}^{a}|(-\Delta+2)^{-1/2}\|,
\end{align*}
where $\mu$ is chosen smaller then $(\inf\sigma(h_{N,Z})-2)$. Clearly, $\mu\leq-(2+Z^2)$ by Remark \ref{rem:h_spec}.

With the aid of the functional calculus, we have
$$\|(-\Delta+2)^{1/2}(-\Delta-\mu)^{-1/2}\|=1$$
and
$$\|(h_{N,Z}-\mu)^{1/2}|h_{N,Z}-\xi|^{-1/2}\|^{2}=\sup_{\lambda\in\sigma(h_{N,Z})}\frac{\lambda-\mu}{|\lambda-\xi|}\leq\max\lsz 1,\frac{-\mu}{d_{N,Z}(\xi)}\psz.$$
To find an upper bound for the norm of $L:=(-\Delta-\mu)^{1/2}(h_{N,Z}-\mu)^{-1/2}$ ($L$ is bounded due to the closed graph theorem), we imitate the proof of \cite[Lemma 4]{DuStTu_10},
\begin{align}\label{eq:quadratic_est}
 \|L\psi\|^{2}&=\langle\psi,L^{*}L\psi\rangle=\|\psi\|^{2}+\langle(h_{N,Z}-\mu)^{-1/2}\psi,(-\mathcal{V}_{\mathrm{2D}})(h_{N,Z}-\mu)^{-1/2}\psi\rangle\nonumber\\
 &\leq\|\psi\|^{2}+\frac{\Gamma(1/4)^{4}}{4\pi^{2}}Z\sqrt{N}\langle(h_{N,Z}-\mu)^{-1/2}\psi,\sqrt{-\Delta}(h_{N,Z}-\mu)^{-1/2}\psi\rangle\\
 &\leq\|\psi\|^{2}+\frac{\Gamma(1/4)^{4}}{4\pi^{2}}Z\sqrt{N}\|(h_{N,Z}-\mu)^{-1/2}\|\|\psi\|\|L\psi\|\nonumber.
\end{align}
Here we made use of \eqref{eq:2d_pot_neg} and the fact that $\sqrt{-\Delta}\leq\sqrt{-\Delta-\mu}$, for $\mu<0$. This implies
$$\|L\|^{2}\leq 1+\frac{\Gamma(1/4)^{4}}{4\pi^{2}}Z\sqrt{N}\|L\|.$$
Therefore, we have 
\begin{equation}\label{eq:free_vs_coul_bound}
\|L\|\leq C_{1}(N,Z):=\frac{1}{8\pi^2}\left(\Gamma\left(\frac{1}{4}\right)^4 Z\sqrt{N}+\sqrt{\Gamma\left(\frac{1}{4}\right)^8 Z^2 N+64\pi^4}\,\right).
\end{equation}

Using Proposition \ref{prop:eff_2d_est}, we conclude that
$$\alpha^{2}\leq\max\lsz 1,\frac{-\mu}{d_{N,Z}(\xi)}\psz C_{1}(N,Z)^{2}B,$$
where $B$ is the right-hand-side of \eqref{eq:eff_2d_est}. 
Moreover, for $a\leq \mathrm{e}^{-1}$, we have
$$B\leq C_{2}(N,Z) a|\ln{a}|$$
with
\begin{equation}\label{eq:pot_diff_bound}
 C_{2}(N,Z):=(2\sqrt{3}+4\sqrt{2})\left[ZN\left(\frac{1}{4}-\frac{1}{\pi^2}\right)+\binom{N}{2}\left(\frac{1}{3}-\frac{5}{4\pi^2}\right)\right].
\end{equation}

For any $\xi\in\mathrm{Res}(h_{N,Z})\cap\R$, there is $a_{0}=a_{0}(\xi)$ such that for all $a$, $0<a<a_{0}(\xi)$, one has $\alpha^2\leq 1/2$. For definiteness, 
we set $a_{0}(\xi)=\min\lsz \mathrm{e}^{-1}, \tilde{a}_{0}(\xi)\psz$, where $\tilde{a}_{0}=\tilde{a}_{0}(\xi)$ is the solution to
\begin{equation}\label{eq:a0_def}
\max\lsz 1,\frac{-\mu}{d_{N,Z}(\xi)}\psz C_{1}(N,Z)^{2}C_{2}(N,Z)\,\tilde{a}_{0}|\ln{\tilde{a}_{0}}|=1/2.
\end{equation}
The assertions of the theorem now follow immediately from Lemma \ref{lem:symm_resolv_eq}.
\end{proof}

\begin{remark}\label{rem:spec_dist}
Under the assumptions of Theorem \ref{theo:eff_2d_coul}, we have
$$\|(H_{\mathrm{eff}}^{a}-NE_{1}^{a}-\xi)^{-1}-(h_{N,Z}-\xi)^{-1}\|\leq\|(h_{N,Z}-\xi)^{-1}\|,$$ 
which, by the functional calculus and the triangle inequality, implies
$$\frac{1}{d_{\mathrm{eff}}(\xi+NE_{1}^{a})}\leq\frac{2}{d_{N,Z}(\xi)}$$
with
\begin{equation}\label{eq:deff_def}
d_{\mathrm{eff}}(\xi):=\mathrm{dist}(\xi,\sigma(H_{\mathrm{eff}}^{a})). 
\end{equation}
\end{remark}

\section{Approximation of the full Hamiltonian by the effective Hamiltonian}
Let us introduce the following notation,
$$Q^{a}:=1-P^{a},\quad  H_{\bot}^{a}:=Q^{a}H_{N,Z}^{a}Q^{a},\quad  R_{\bot}^{a}(\xi):=(H_{\bot}^{a}-\xi)^{-1}.$$
$H_{\bot}^{a}$ is well defined on $\mathrm{Dom}(H_{N,Z}^{a})$, because $\mathrm{Dom}(H_{N,Z}^{a})$ is invariant under $Q^{a}$. In what follows we will view $H_{\bot}^{a}$ as an operator acting in $\mathrm{Ran}Q^{a}$ with domain $Q^{a}\mathrm{Dom}(H_{N,Z}^{a})$.
Furthermore, denote
$$\mathscr{W}^{a}(\xi):=P^{a}\mathcal{V}Q^{a}R_{\bot}^{a}(\xi)Q^{a}\mathcal{V}P^{a},\quad R_{\text{eff}}^{\mathscr{W}}(\xi):=\left(H_{\text{eff}}^{a}-\mathscr W^{a}(\xi)-\xi\right)^{-1},$$
where
$$\mathcal{V}:=\mathcal{V}_{en}+\mathcal{V}_{ee}$$
with
$$\mathcal{V}_{en}:=-\sum_{i=1}^{N}\frac{Z}{r_{i}},\quad \mathcal{V}_{ee}:=\sum_{1\leq i<j\leq N}\frac{1}{r_{i,j}}.$$

With respect to the decomposition $L^{2}(\Omega_{a})^{\otimes^{N}}=\mathrm{Ran}P^{a}\oplus\mathrm{Ran}Q^{a}$, we have
\begin{equation*}
H_{N,Z}^{a}=\begin{pmatrix}
	    H_{\mathrm{eff}}^{a}&P^{a}H_{N,Z}^{a}Q^{a}\\
	    Q^{a}H_{N,Z}^{a}P^{a}&H_{\bot}^{a}
	    \end{pmatrix}=
	    \begin{pmatrix}
	    H_{\mathrm{eff}}^{a}&P^{a}\mathcal{V}Q^{a}\\
	    Q^{a}\mathcal{V}P^{a}&H_{\bot}^{a}.
	    \end{pmatrix}
\end{equation*}
The second equality follows from the fact that $P^{a}$ commutes with $-\Delta_{\Omega_{a}^{N}}$. By direct inspection one arrives at the so-called Feshbach formula,
\begin{equation}\label{eq:Feshbach}
(H_{N,Z}^{a}-\xi)^{-1}=\begin{pmatrix}
	     R_{\text{eff}}^{\mathscr{W}} & -R_{\text{eff}}^{\mathscr{W}}P^{a}\mathcal{V}Q^{a}R_{\bot}^{a}\\
	     -R_{\bot}^{a}Q^{a}\mathcal{V}P^{a}R_{\text{eff}}^{\mathscr{W}} &	 R_{\bot}^{a}+R_{\bot}^{a}Q^{a}\mathcal{V}P^{a}R_{\text{eff}}^{\mathscr{W}}P^{a}\mathcal{V}Q^{a}R_{\bot}^{a}
             \end{pmatrix},
\end{equation}
which holds for those $\xi\in\C$ such that $R_{\bot}^{a}(\xi)$ and $R_{\text{eff}}^{\mathscr{W}}(\xi)$ exist and are bounded on $\mathrm{Ran}Q^{a}$ and $\mathrm{Ran}P^{a}$, respectively.

From now on, consider $N\geq 2$.

\begin{proposition}\label{prop:tot_eff_1}
Let $0<a<C_{3}(N,Z)$, where
\begin{equation*}
 C_{3}(N,Z):=\sqrt{3}\pi/[2N(N-1+2Z)],
\end{equation*}
 $\xi<NE_{1}^{a}$, and $\xi\notin\sigma(H_{\mathrm{eff}}^{a}-\mathscr{W}^{a}(\xi))$. Then  $\xi\in\mathrm{Res}(H_{N,Z}^{a})$ and
 $$\|(H_{N,Z}^{a}-\xi)^{-1}-R_{\mathrm{eff}}^{\mathscr{W}}\oplus 0\|\leq\frac{a}{C_{3}(N,Z)~d_{\mathrm{eff}}^{\mathscr{W}}(\xi)}\left(1+\frac{a}{C_{3}(N,Z)}\right)+\frac{2a^{2}}{3\pi^{2}}.$$
where
$$d_{\mathrm{eff}}^{\mathscr{W}}(\xi):=\mathrm{dist}\big(\xi,\sigma(H_{\mathrm{eff}}^{a}-\mathscr{W}^{a}(\xi))\big).$$
\end{proposition}
\begin{proof}
The  proof is strongly inspired by a similar one in \cite{BrDu_06}.

With the help of the following formula,
$$\left\|\begin{pmatrix}0 & A\\
A^{\dagger} & 0\end{pmatrix}\right\|^{2}=\|AA^{\dagger}\|=\|A\|^{2},$$
one derives that
\begin{equation}\label{eq:Feshbach_est}
\begin{split}
\|(H_{N,Z}^{a}-\xi)^{-1}-R_{\text{eff}}^{\mathscr{W}}(\xi)&\oplus0\|\leq\|R_{\text{eff}}^{\mathscr{W}}P^{a}\mathcal{V}Q^{a}R_{\bot}^{a}\|+\|R_{\bot}^{a}Q^{a}\mathcal{V}P^{a}R_{\text{eff}}^{\mathscr{W}}P^{a}\mathcal{V}Q^{a}R_{\bot}^{a}\|\\
&+\|R_{\bot}^{a}\|\leq\frac{1}{d_{\mathrm{eff}}^{\mathscr{W}}}\,\|\mathcal{V}Q^{a}R_{\bot}^{a}\|\left(1+\|\mathcal{V}Q^{a}R_{\bot}^{a}\|\right)+\|R_{\bot}^{a}\|.
\end{split}
\end{equation}

Since 
$$T_{\bot}:=Q^{a}(-\Delta_{\Omega_{a}^{N}})Q^{a}\geq (N-1)E_{1}^{a}+E_{2}^{a}=NE_{1}^{a}+\frac{3\pi^{2}}{a^{2}},$$
we have
\begin{equation}\label{eq:R_0_bound}
0\leq R_{0}:=(T_{\bot}-\xi)^{-1}\leq\frac{a^{2}}{3\pi^{2}}.
\end{equation}
Further, let us estimate $\big\|\mathcal{V}Q^{a}R_{0}^{\,1/2}\big\|=\big\|R_{0}^{\,1/2}Q^{a}\mathcal{V}^{2}Q^{a}R_{0}^{\,1/2}\big\|^{1/2}$. Since $\mathcal{V}^{2}\leq\mathcal{V}_{en}^{2}+\mathcal{V}_{ee}^{2}$, we can can find bounds for the $en$ and $ee$-terms separately.

{\it Bound for $\big\|\mathcal{V}_{en}Q^{a}R_{0}^{\,1/2}\big\|$}: The following estimate,
$$\mathcal{V}_{en}^{2}\leq \frac{Z^{2}}{2}\sum_{i,j}\Big(\frac{1}{r_{i}^{2}}+\frac{1}{r_{j}^{2}}\Big)=Z^{2}N\sum_{i}\frac{1}{r_{i}^{2}}$$
together with the Hardy inequality \eqref{eq:Hardy_ineq},
implies
\begin{align*}
R_{0}^{\,1/2}Q^{a}\mathcal{V}_{en}^{2}Q^{a}R_{0}^{\,1/2}&\leq 4Z^{2}N~R_{0}^{\,1/2}Q^{a}(-\Delta_{\Omega_{a}^{N}})Q^{a}R_{0}^{\,1/2}=4Z^{2}N(Q^{a}+\xi R_{0})\\
&\leq 4Z^{2}N(1+N/3)\leq 4Z^{2}N^{2},
\end{align*}
whenever $N\geq 2$, and so
\begin{equation}\label{eq:Ven_bound}
\big\|\mathcal{V}_{en}Q^{a}R_{0}^{\,1/2}\big\|\leq 2ZN.
\end{equation}

{\it Bound for $\big\|\mathcal{V}_{ee}Q^{a}R_{0}^{\,1/2}\big\|$:} In \cite[Lemma 3.2]{BrDu_06} it was deduced directly from the Hardy inequality that
$$\frac{1}{r_{i,j}^{2}}\leq 2(-\Delta_{\bm{r}_{i}}-\Delta_{\bm{r}_{j}}).$$
The same holds true on $\mathcal{H}^1_{0}(\Omega_{a})^{\otimes 2}$. Using this result we have
$$\mathcal{V}_{ee}^{2}\leq\frac{1}{2}\sum_{i<j,\,k<l}\Big(\frac{1}{r_{i,j}^{2}}+\frac{1}{r_{k,l}^{2}}\Big)\leq N(N-1)\sum_{i<j}(-\Delta_{\bm{r}_{i}}-\Delta_{\bm{r}_{j}})=N(N-1)^{2}(-\Delta_{\Omega_{a}^{N}}).$$
Consequently, in the same manner as for the $en$-term,
\begin{equation}\label{eq:Vee_bound}
\big\|\mathcal{V}_{ee}Q^{a}R_{0}^{\,1/2}\big\|\leq N(N-1).
\end{equation}

{\it Bound for $\|R_{\bot}^{a}\|$:} From \eqref{eq:R_0_bound}, \eqref{eq:Ven_bound}, and \eqref{eq:Vee_bound}, it follows
\begin{align*}
 &(R_{0}^{1/2}Q^{a}\mathcal{V}_{en}Q^{a}R_{0}^{1/2})^{2}=R_{0}^{1/2}Q^{a}\mathcal{V}_{en}Q^{a}R_{0}Q^{a}\mathcal{V}_{en}Q^{a}R_{0}^{1/2}\leq\frac{1}{3}\left(\frac{2ZNa}{\pi}\right)^{2},\\
 &(R_{0}^{1/2}Q^{a}\mathcal{V}_{ee}Q^{a}R_{0}^{1/2})^{2}\leq\frac{1}{3}\left(\frac{N(N-1)a}{\pi}\right)^{2}.
\end{align*}
Thus we have
$$\|R_{0}^{1/2}Q^{a}\mathcal{V}Q^{a}R_{0}^{1/2}\|\leq \frac{2ZNa}{\sqrt{3}\pi}+\frac{N(N-1)a}{\sqrt{3}\pi}=\frac{Na}{\sqrt{3}\pi}(N-1+2Z).$$
For $a$ small enough, this bound is smaller then one, and so by the symmetrized resolvent formula,
\begin{equation}\label{eq:r_bot_bound}
R_{\bot}^{a}(\xi)=(T_{\bot}+Q^{a}\mathcal{V}Q^{a}-\xi)^{-1}=R_{0}^{\,1/2}\left(1+R_{0}^{\,1/2}Q^{a}\mathcal{V}Q^{a}R_{0}^{\,1/2}\right)^{-1}R_{0}^{\,1/2},
\end{equation}
one has $\xi\in\mathrm{Res}(H_{\bot}^{a})$ and that the resolvent $R_{\bot}^{a}(\xi)$
is positive. Moreover,
$$\|R_{\bot}^{a}(\xi)\|\leq\frac{\|R_{0}\|}{1-\big\|R_{0}^{\,1/2}Q^{a}\mathcal{V}Q^{a}R_{0}^{\,1/2}\big\|}\,.$$
For $a<\sqrt{3}\pi/[2N(N-1+2Z)]$,
\begin{equation}\label{eq:Rbot_bound}
\|R_{\bot}^{a}\|\leq2\|R_{0}\|\leq\frac{2a^{2}}{3\pi^{2}}\,.
\end{equation}

{\it Bound for $\|\mathcal{V}Q^{a}R_{\bot}^{a}\|$:} With the help of \eqref{eq:r_bot_bound},
$$\|\mathcal{V}Q^{a}R_{\bot}^{a}\|\leq\frac{\|\mathcal{V}Q^{a}R_{0}^{\,1/2}\|\,\|R_{\bot}^{a}\|^{1/2}}{\left(1-\big\|R_{0}^{\,1/2}Q^{a}\mathcal{V}Q^{a}R_{0}^{\,1/2}\big\|\right)^{\!1/2}}\leq\frac{2Na}{\sqrt{3}\pi}(N-1+2Z),$$
where we used \eqref{eq:Ven_bound}, \eqref{eq:Vee_bound}, and \eqref{eq:Rbot_bound}.
\end{proof}

The following lemma is an extension of its single-electron version \cite[Lemma 11]{DuStTu_10}.

\begin{lemma}\label{lem:W_est}
 Let $0<a<C_{3}(N,Z)$. If $\xi<NE_{1}^{a}$, then $\mathscr{W}^{a}(\xi)$ is positive and
$$\|(-\Delta+\alpha)^{-1/2}\mathscr{W}^{a}(\xi)(-\Delta+\alpha)^{-1/2}\|\leq \frac{\Gamma(1/4)^{4}N^{3/2}}{6\pi^{3}~\sqrt{\alpha}}\left(Z^{2}+\frac{(N-1)^{2}}{\sqrt{2}}\right)\,a,$$
for any $\alpha>0$.
\end{lemma}
\begin{proof}
In course of the proof of Proposition \ref{prop:tot_eff_1} we demonstrated that under the assumptions of the lemma, $R_{\bot}^{a}(\xi)$ is positive and so is $\mathscr{W}^{a}(\xi)$. 
Moreover, using \eqref{eq:Rbot_bound} we get
\begin{align*}
 0\leq\mathscr{W}^{a}(\xi)&=P^{a}\mathcal{V}Q^{a}R_{\bot}^{a}Q^{a}\mathcal{V}P^{a}\leq\frac{2a^{2}}{3\pi^{2}}P^{a}\mathcal{V}^{2}P^{a}\leq\frac{2a^{2}}{3\pi^{2}}P^{a}(\mathcal{V}_{en}^{2}+\mathcal{V}_{ee}^{2})P^{a}\\
 &\leq \frac{2a^{2}}{3\pi^{2}}\left(Z^{2}N\sum_{i=1}^{N}P^{a}\frac{1}{r_{i}^{2}}P^{a}+\binom{N}{2}\sum_{1\leq i<j\leq N}P^{a}\frac{1}{r_{i,j}^{2}}P^{a}\right).
\end{align*}
Since
\begin{align*}
 P^{a}\frac{1}{r_{i,j}^{2}}P^{a}&=\frac{4}{a^{2}}\int\limits_{-a/2}^{a/2}\int\limits_{-a/2}^{a/2}\frac{\cos^{2}\frac{\pi s}{a}\cos^{2}\frac{\pi t}{a}}{\varrho_{i,j}^{2}+(s-t)^{2}}\dd s\dd t\leq\frac{4}{a^{2}}\int\limits_{-a/2}^{a/2}\int\limits_{-a/2-t}^{a/2-t}\frac{\dd s\dd t}{\varrho_{i,j}^{2}+s^{2}}\\
 &\leq\frac{8}{a}\int\limits_{0}^{\infty}\frac{\dd s}{\varrho_{i,j}^{2}+s^{2}}=\frac{4\pi}{a\varrho_{i,j}}
\end{align*}
and similarly
$$P^{a}\frac{1}{r_{i}^{2}}P^{a}\leq\frac{2\pi}{a\varrho_{i}},$$
we conclude that
\begin{align*}
 \mathscr{W}^{a}(\xi)&\leq\frac{4Na}{3\pi}\left(Z^{2}\sum_{i}\frac{1}{\varrho_{i}}+(N-1)\sum_{i<j}\frac{1}{\varrho_{i,j}}\right)\\
 &\leq\frac{\Gamma(1/4)^{4}Na }{3\pi^{3}}\left(Z^{2}+\frac{(N-1)^{2}}{\sqrt{2}}\right)\langle\hat{\psi},\sum_{i}|\bm{\lambda}_{i}|\hat{\psi}\rangle\\
 &\leq\frac{\Gamma(1/4)^{4}N^{3/2}a }{3\pi^{3}}\left(Z^{2}+\frac{(N-1)^{2}}{\sqrt{2}}\right)\langle\psi,\sqrt{-\Delta}\,\psi\rangle,
 \end{align*}
in the same manner as in the proof of Proposition \ref{prop:klmn_2d}.
The lemma now readily follows, since
$$\|(-\Delta+\alpha)^{-1/2}\sqrt{-\Delta}(-\Delta+\alpha)^{-1/2}\|=\sup_{\lambda\in[0,\infty)}\frac{\sqrt{\lambda}}{\lambda+\alpha}=\frac{1}{2\sqrt{\alpha}}.$$
\end{proof}

\begin{proposition}\label{prop:tot_eff_2}
Suppose that $\xi\in\mathrm{Res}(H_{\mathrm{eff}}^{a})\cap\mathbb{R}$ and set  
\begin{equation}\label{eq:mu_set}
\mu=-N\left(\frac{\Gamma(1/4)^4Z}{8\pi^2}\right)^2-2.
\end{equation}	
If $a\leq a_{1}(\xi)$, with $a_{1}(\xi)$ given by \eqref{eq:a_1}, then $\xi\notin\sigma(H_{\mathrm{eff}}^{a}-\mathscr{W}^{a}(\xi))$
and \begin{equation*}
\|R_{\mathrm{eff}}^{\mathscr{W}}(\xi)-(H_{\mathrm{eff}}^{a}-\xi)^{-1}\|\leq\frac{C_{4}(N,Z)}{d_{\mathrm{eff}}(\xi)}\max\left\lbrace 1\,,\,\frac{-\mu}{d_{\mathrm{eff}}(\xi)}\right\rbrace\,a,
\end{equation*}
where 
\begin{equation*}
 C_{4}(N,Z):=2C_{1}(N,Z)^2 \frac{\Gamma(1/4)^{4}N^{3/2}}{6\pi^{3}~\sqrt{-\mu}}\left(Z^{2}+\frac{(N-1)^{2}}{\sqrt{2}}\right)
\end{equation*}
and $d_{\mathrm{eff}}(\xi)$ is defined by \eqref{eq:deff_def}.
\end{proposition}
\begin{proof}
Due to Remark \ref{rem:eff_spectrum}, $H_{\mathrm{eff}}^{a}-NE_{1}^{a}-\mu>0$ and $\xi<NE_{1}^{a}$.

We will proceed as in the proof of Theorem \ref{theo:eff_2d_coul}. Apply Lemma \ref{lem:symm_resolv_eq} with $A=H_{\mathrm{eff}}^{a}-\xi$,
$C=-\mathscr{W}^{a}(\xi)$. We have
\begin{equation}
\begin{split}\label{eq:alpha_est}
 \alpha^{2} & =\big\||H_{\mathrm{eff}}^{a}-\xi|^{-1/2}\mathscr{W}^{a}|H_{\mathrm{eff}}^{a}-\xi|^{-1/2}\big\|\\
 & \leq\big\|(-\Delta-\mu)^{-1/2}\mathscr{W}^{a}(-\Delta-\mu)^{-1/2}\big\|\,\|(-\Delta-\mu)^{1/2}(H_{\mathrm{eff}}^{a}-NE_{1}^{a}-\mu)^{-1/2}\|^{2}\\
 & \hspace{1em}\,\times\,\big\|(H_{\mathrm{eff}}^{a}-NE_{1}^{a}-\mu)^{1/2}|H_{\mathrm{eff}}^{a}-\xi|^{-1/2}\big\|^{2}.
\end{split}
\end{equation}
Note that $\mathscr{W}^{a}(\xi)$ is positive under the assumptions of Lemma \ref{lem:W_est}.

The upper bound for $\tilde{L}:=(-\Delta-\mu)^{1/2}(H_{\mathrm{eff}}^{a}-NE_{1}^{a}-\mu)^{-1/2}$ is the same as that for the operator $L$ in the proof of Theorem \ref{theo:eff_2d_coul}. Indeed,
since 
$$-P^{a}\mathcal{V}P^{a}\leq P^{a}\mathcal{V}_{en}P^{a}\leq\sum_{i=1}^{N}\frac{Z}{\varrho_{i}}$$
we arrive at \eqref{eq:quadratic_est} with $L$ replaced by $\tilde{L}$, $\mathcal{V}_{\mathrm{2D}}$ by $P^{a}\mathcal{V}P^{a}$, and $h_{N,Z}$ by $H_{\mathrm{eff}}^{a}$.
Consequently,
\begin{equation}\label{eq:l_tilde_est}
\|\tilde{L}\|=\|(-\Delta-\mu)^{1/2}(H_{\mathrm{eff}}^{a}-NE_{1}^{a}-\mu)^{-1/2}\|\leq C_{1}(N,Z).
\end{equation}
Furthermore we observe that
\begin{equation}\label{eq:h_eff_bound}
\big\|(H_{\mathrm{eff}}^{a}-NE_{1}^{a}-\mu)^{1/2}|H_{\mathrm{eff}}^{a}-\xi|^{-1/2}\big\|^{2}\leq\max\left\lbrace 1\,,\,\frac{-\mu}{d_{\mathrm{eff}}(\xi)}\right\rbrace
\end{equation}
by the functional calculus.

Putting \eqref{eq:alpha_est}, \eqref{eq:l_tilde_est}, \eqref{eq:h_eff_bound}, and Lemma \ref{lem:W_est} together, we deduce that there exists a positive $a_{1}=a_{1}(\xi)$ such that if $a\leq a_{1}(\xi)$, then $\alpha^{2}\leq 1/2$. Lemma \ref{lem:symm_resolv_eq} now gives $\xi\notin\sigma(H_{\mathrm{eff}}^{a}-\mathscr{W}^{a}(\xi))$ and 
\begin{equation}\label{eq:eff_W_diff}
\|R_{\mathrm{eff}}^{\mathscr{W}}(\xi)-(H_{\mathrm{eff}}^{a}-\xi)^{-1}\|\leq\frac{2\alpha^{2}}{d_{\mathrm{eff}}(\xi)}.
\end{equation}

We make $a_{1}(\xi)$ definite by setting
\begin{equation}\label{eq:a_1}
a_{1}(\xi):=\min\{C_{3}(N,Z),\ \tilde{a}_1(\xi)\},
\end{equation}
where $\tilde{a}_1(\xi)$ is the unique solution of
\begin{equation}\label{eq:tilde_a_1}
\frac{\Gamma(1/4)^{4}N^{3/2}}{6\pi^{3}~\sqrt{-\mu}}\left(Z^{2}+\frac{(N-1)^{2}}{\sqrt{2}}\right)\max\lsz 1,\frac{-\mu}{d_{\mathrm{eff}}(\xi)}\psz C_{1}(N,Z)^2\, \tilde{a}_{1}=\frac{1}{2}.
\end{equation}
\end{proof}

\begin{theorem}\label{theo:total_eff_comp}
 Let $\xi\in\mathrm{Res}(H_{\mathrm{eff}}^{a})\cap\R$. If $a<a_{1}(\xi)$,
 where $a_{1}(\xi)$ is given by \eqref{eq:a_1}(with $\mu$ introduced in \eqref{eq:mu_set}), then $\xi\in\mathrm{Res}(H_{N,Z}^{a})$ and
\begin{equation}\label{eq:total_eff_bound}
\begin{split} 
\|(H_{N,Z}^{a}-\xi)^{-1}&-(H_{\mathrm{eff}}^{a}-\xi)^{-1}\oplus 0\|\\
&\leq\frac{1}{d_{\mathrm{eff}}(\xi)}\left[4C_{3}(N,Z)^{-1}+C_{4}(N,Z)\max\lsz 1,\frac{-\mu}{d_{\mathrm{eff}}(\xi)}\psz\right]\, a
 +\frac{2a^2}{3\pi^2}.
\end{split}
\end{equation}
\end{theorem} 
\begin{proof}
We may apply Proposition \ref{prop:tot_eff_2} that yields $\xi\notin\sigma(H_{\mathrm{eff}}^{a}-\mathscr{W}^{a}(\xi))$. So the assumptions of Proposition \ref{prop:tot_eff_1} are fulfilled too. Thus $\xi\in\mathrm{Res}(H^{a}_{N,Z})$. Furthermore, \eqref{eq:eff_W_diff} holds with $\alpha^2<1/2$, which implies
$$\frac{1}{d_{\mathrm{eff}}^{\mathscr{W}}(\xi)}\leq\frac{2}{d_{\mathrm{eff}}(\xi)}.$$ Therefore, we arrive at the following estimate
\begin{equation*}
 \begin{split}
 \|(H^{a}_{N,Z}-\xi)^{-1}&-(H_{\mathrm{eff}}^{a}-\xi)^{-1}\oplus0\|\\
 &\leq\|(H_{N,Z}^{a}-\xi)^{-1}-R_{\text{eff}}^{\mathscr{W}}(\xi)\oplus0\|+\|R_{\mathrm{eff}}^{\mathscr{W}}(\xi)-(H_{\mathrm{eff}}^{a}-\xi)^{-1}\|\\
 & \leq\frac{1}{d_{\mathrm{eff}}(\xi)}\left[4C_{3}(N,Z)^{-1}+C_{4}(N,Z)\max\lsz 1,\frac{-\mu}{d_{\mathrm{eff}}(\xi)}\psz\right]\, a+\frac{2a^2}{3\pi^2}.
 \end{split}
\end{equation*}
\end{proof}

\section{Approximation of the total Hamiltonian by the two-dimensional atomic Hamiltonian}\label{sec:main}
\begin{theorem} \label{theo:main}
Let $\xi\in\mathrm{Res}(h_{N,Z}+NE_{1}^{a})\cap\R$. Set $\mu$ as in \eqref{eq:mu_set}. If $a>0$ fulfills
$$a<a_{3}(\xi):=\min\lsz\mathrm{e}^{-1},\,C_{3}(N,Z),\,\tilde{a}_{0}(\xi-NE_{1}^{a}),\,a_{2}(\xi)\psz,$$
where $\tilde{a}_{0}(\xi)$ is defined by \eqref{eq:a0_def} and $a_{2}=a_{2}(\xi)$ is the solution to
$$C_{4}(N,Z)\max\lsz 1,\frac{-2\mu}{d_{N,Z}(\xi-NE_{1}^{a})}\psz a_{2}=1,$$
then $\xi\in\mathrm{Res}(H_{N,Z}^{a})$ and
\begin{equation*}
 \begin{split}
&\|(H_{N,Z}^{a}-\xi)^{-1}-(h_{N,Z}+NE_{1}^{a}-\xi)^{-1}\oplus 0\|\\
&\leq \frac{2}{d_{N,Z}(\xi-NE_{1}^{a})}\max\lsz 1,\frac{-\mu}{d_{N,Z}(\xi-NE_{1}^{a})}\psz C_{1}(N,Z)^{2}C_{2}(N,Z)\,a|\ln{a}|\\
&+\frac{2}{d_{N,Z}(\xi-NE_{1}^{a})}\left[4C_{3}(N,Z)^{-1}+C_{4}(N,Z)\max\lsz 1,\frac{-2\mu}{d_{N,Z}(\xi-NE_{1}^{a})}\psz\right]a+\frac{2a^2}{3\pi^2}.
 \end{split}
\end{equation*}
\end{theorem}

\begin{proof}
Due to the bound on $a$, we may apply Theorem \ref{theo:eff_2d_coul} with $\xi-NE_{1}^{a}$ substituted for $\xi$. It yields  $\xi\in\mathrm{Res}(H_{\mathrm{eff}}^{a})$ and
\begin{equation}\label{eq:eff_2d_coul_shifted}
 \begin{split}
\|(H_{\mathrm{eff}}^{a}&-\xi)^{-1}-(h_{N,Z}+NE_{1}^{a}-\xi)^{-1}\|\\
&\leq\frac{2}{d_{N,Z}(\xi-NE_{1}^{a})}\max\lsz 1,\frac{-\mu}{d_{N,Z}(\xi-NE_{1}^{a})}\psz C_{1}(N,Z)^{2}C_{2}(N,Z)\,a|\ln{a}|.
 \end{split}
\end{equation}
Moreover, by Remark \ref{rem:spec_dist},
\begin{equation}\label{eq:spec_dist}
\frac{1}{d_{\mathrm{eff}}(\xi)}\leq \frac{2}{d_{N,Z}(\xi-NE_{1}^{a})}. 
\end{equation}
Therefore, $a_{2}(\xi)\leq \tilde{a}_{1}(\xi)$, where $\tilde{a}_{1}(\xi)$ is given by \eqref{eq:tilde_a_1}, and so the assumptions of Theorem \ref{theo:total_eff_comp} are also fulfilled. Thus we have
$\xi\in\mathrm{Res}(H_{N,Z}^{a})$.
Observe that
\begin{equation}
\begin{split}
 \|(H_{N,Z}^{a}-\xi)^{-1}-(h_{N,Z}+NE_{1}^{a}-&\xi)^{-1}\oplus 0\|\leq\|(H_{N,Z}^{a}-\xi)^{-1}-(H_{\mathrm{eff}}^{a}-\xi)^{-1}\oplus 0\|\\
 &+\|(H_{\mathrm{eff}}^{a}-\xi)^{-1}-(h_{N,Z}+NE_{1}^{a}-\xi)^{-1}\|.
\end{split}
\end{equation}
Putting this together with \eqref{eq:total_eff_bound}, \eqref{eq:eff_2d_coul_shifted}, and \eqref{eq:spec_dist} finishes the proof.
\end{proof}

\begin{remark}\label{rem:distance}
 If we set $\xi=NE_{1}^{a}+\delta$ with some fixed $\delta\in\mathrm{Res}(h_{N,Z})$, then $a_{3}$ does not depend on the parameter $a$. In fact, it depends only on $d_{N,Z}(\delta)$.
\end{remark}

\section{Properties of the eigenvalues}\label{sec:spectrum}

\subsection{Localization}\label{sec:localization}
Suppose that there is an isolated eigenvalue, say $\lambda$, of $h_{N,Z}$ with a finite multiplicity. 
If we set $\xi_{+}:=NE_{1}^{a}+\lambda+d$ with $0<d<\mathrm{dist}(\lambda,\sigma(h_{N,Z})\setminus\lsz\lambda\psz)/2$, then
$d_{N,Z}(\xi_{+}-NE_{1}^{a})=d$ and in the view of Theorem \ref{theo:main} and Remark \ref{rem:distance} there exists $a_{\mathrm{min}}(d)>0$ such that
for all $a<a_{\mathrm{min}}(d)$ we have
\begin{equation}\label{eq:loc_est}
 \|(H_{N,Z}^{a}-\xi_{+})^{-1}-(h_{N,Z}+NE_{1}^{a}-\xi_{+})^{-1}\oplus 0\|\leq K(d) a|\ln{a}|.
\end{equation}
Let us stress that the value of $K(d)\in\R_{+}$, as well as that of $a_{\mathrm{min}}(d)$, depends only on $d,\, N,$ and $Z$, but not on the particular eigenvalue $\lambda$ or the value of $a$.

Furthermore, let $\Gamma$ stands for the anti-clockwise oriented circle with center $NE_{1}^{a}+\lambda$ and radius $d$. With the aid of formula (3.10) in \cite[Chpt. IV]{kato} we can propagate \eqref{eq:loc_est} to all $\xi\in\Gamma$, 
\begin{equation*}
 \|(H_{N,Z}^{a}-\xi)^{-1}-(h_{N,Z}+NE_{1}^{a}-\xi)^{-1}\oplus 0\|\leq\frac{9 K(d)\,a|\ln{a}|}{1-6dK(d)\,a|\ln{a}|},
\end{equation*}
for $a$ small enough so that $6dK(d)\,a|\ln{a}|<1$.
Consequently we arrive at the following estimate for the difference of the projections $P_{\Gamma}$ and $p_{\Gamma}$ onto the spectrum of $H_{N,Z}^{a}$ and $h_{N,Z}+NE_{1}^{a}$, respectively, inside $\Gamma$,
\begin{equation}\label{eq:proj_dif_bound}
 \begin{split}
  \|P_{\Gamma}-p_{\Gamma}\oplus 0\|&=\frac{1}{2\pi}\Big\|\int_{\Gamma}(H_{N,Z}^{a}-\xi)^{-1}-(h_{N,Z}+NE_{1}^{a}-\xi)^{-1}\oplus 0\,\dd\xi\Big\|\\
 &\leq \frac{9dK(d)\,a|\ln{a}|}{1-6dK(d)\,a|\ln{a}|}.
 \end{split}
\end{equation}
The right-hand-side of \eqref{eq:proj_dif_bound} is strictly increasing on some sufficiently small right neighborhood of $0$ and it tends to zero as $a\to 0$. Consequently, $\tilde{a}_{\mathrm{min}}(d)$ exists, $0<\tilde{a}_{\mathrm{min}}(d)\leq a_{\mathrm{min}}(d)$, such that for all 
$a<\tilde{a}_{\mathrm{min}}(d)$,
$$\|P_{\Gamma}-p_{\Gamma}\oplus 0\|<1.$$
Therefore, for these values of $a$, in the $d$-neighborhood of $(\lambda+NE_{1}^{a})$ there is the exactly same number of eigenvalues (counting multiplicity) of $H_{N,Z}^{a}$ as the multiplicity of $\lambda$ in the spectrum of  $h_{N,Z}$ is.

The idea above may be applied on a finite cluster of successive eigenvalues, $\lambda_{1}\,\dots,\,\lambda_{M}$, of $h_{N,Z}$. We just take $d$ sharply smaller than a half of the minimum of isolation distances of all $\lambda_{i}$. Moreover we may perform  similar estimates as above on intervals $[\lambda_{i}+d,\lambda_{i+1}-d]$ (more concretely, we change $\lambda$ for $(\lambda_{i}+\lambda_{i+1})/2$ and $d$ for $(\lambda_{i+1}-\lambda_{i})/2-d$) to conclude that, for all $a$ small enough, there are no  eigenvalues of $H_{N,Z}^{a}$ in these intervals. 

\subsection{Analyticity}
Consider a unitary mapping $U_{a}:L^{2}(\Omega_{a})^{\otimes^N}\to L^{2}(\Omega_{1})^{\otimes^N}$ given by
$$(U_{a}\psi)(\mathbf{x}_{1},\ldots,\mathbf{x}_{N})=a^{3N/2}\psi(a\mathbf{x}_{1},\ldots,a\mathbf{x}_{N}),$$
 then
$$U_{a}H_{N,Z}^{a}U_{a}^{\dagger}=\frac{1}{a^{2}}\left(-\Delta_{\Omega_{1}^{N}}-\sum_{i=1}^{N}\frac{aZ}{r_{i}}+\sum_{i<j}^{N}\frac{a}{r_{i,j}}\right)=\frac{1}{a^2}(-\Delta_{\Omega_{1}^{N}}+a\mathcal{V})=:\tilde{H}_{N,Z}^{a}.$$
Observe that

 $\bullet$ For all $a>0$,  $\mathrm{Dom}(\tilde{H}_{N,Z}^{a})=\big(\mathcal{H}_{0}^{1}(\Omega_{1})\cap\mathcal{H}^{2}(\Omega_{1})\big)^{\otimes^N}=\mathrm{Dom}(-\Delta_{\Omega_{1}^{N}})=:\mathscr{D}$. We can also extend the definition of $\tilde{H}_{N,Z}^{a}$ to all $a\in\C\setminus\lsz 0\psz$. The resulting operator is well defined on $\mathscr{D}$ due to the Hardy inequality.
 
 $\bullet$ For all $a\in\C\setminus\{ 0\}$ and $\psi\in\mathscr{D}$, $\tilde{H}_{N,Z}^{a}\psi$ has a derivative with respect to $a$,
  and so $a\mapsto \tilde{H}_{N,Z}^{a}$ is analytic in $\C\setminus\{ 0\}$.
  
 $\bullet$ For $a\in\C$, \eqref{eq:full_pot_bound} implies that
 \begin{equation}\label{eq:inf_boundness}
 \|a\mathcal{V}\|\leq|a|\left(2Z^{2}N+N(N-1)^2\right)^{1/2}(\epsilon\|-\Delta_{\Omega_{1}^{N}}\psi\|+\epsilon^{-1}\|\psi\|).
 \end{equation}
 Thus $a\mathcal{V}$ is infinitesimally $(-\Delta_{\Omega_{1}^{N}})$-bounded. Since $-\Delta_{\Omega_{1}^{N}}$ is closed, the same holds true for $\tilde{H}_{N,Z}^{a}$\cite[Thm. IV.1.1]{kato}.
 
 $\bullet$ From \eqref{eq:inf_boundness} it follows that
 $$\|a\mathcal{V}\|\leq |a|\left(2Z^{2}N+N(N-1)^2\right)^{1/2}\left(\epsilon\|(-\Delta_{\Omega_{1}^{N}}-\xi)\psi\|+(\epsilon^{-1}+\epsilon|\xi|)\|\psi\|\right).$$
 If $\xi<NE_{1}^{1}=N\pi^{2}$ then $\xi\in\mathrm{Res}(-\Delta_{\Omega_{1}^{N}})$. Theorem IV.1.16 in \cite{kato} says that $\xi\in\mathrm{Res}(-\Delta_{\Omega_{1}^{N}}+a\mathcal{V})$ whenever
 $$|a|\left(2Z^{2}N+N(N-1)^2\right)^{1/2}\left((\epsilon^{-1}+\epsilon|\xi|)\|(-\Delta_{\Omega_{1}^{N}}-\xi)^{-1}\|+\epsilon\right)<1.$$
 Since $\|(-\Delta_{\Omega_{1}^{N}}-\xi)^{-1}\|=(N\pi^2-\xi)^{-1}$ this can be achieved with $\epsilon=|\xi|^{-1/2}$ and $\xi$ sufficiently negative. Thus the resolvent set of $\tilde{H}_{N,Z}^{a}$ is non-empty for all $a\in\C\setminus\{ 0\}$. 

So here comes the main result of this subsection.
\begin{proposition}
 $\tilde{H}_{N,Z}^{a}$ forms an analytic family of type (A) on $\C\setminus\lsz 0\psz$ and consequently it forms an analytic family in the sense of Kato, see  \cite[Thm. XII.9]{rs4}. 
\end{proposition}

Consequently the analyticity statement \cite[Thm. XII.8]{rs4}  holds for the non-degenerate isolated eigenvalues of the operator $H_{N,Z}^{a}$. In particular it says that if, for some $a_{0}>0$, there is an non-degenerate isolated eigenvalue of $H_{N,Z}^{a_{0}}$, then for $a$ near $a_{0}$ there is exactly one isolated non-degenerate eigenvalue of $H_{N,Z}^{a}$ near this eigenvalue of $H_{N,Z}^{a_{0}}$. 

\begin{remark}[Monotonicity of the eigenvalues]
 In the exactly same manner as in \cite{DuHo_10}, i.e., employing the min-max principle, one can prove that the eigenvalues of $H_{N,Z}^{a}$ (if there are some) are strictly decreasing functions of $a$.
\end{remark}

\section{Reduction to the fermionic subspace}\label{sec:fermionized}
As the physical electrons are fermions, we should reduce $H^{a}_{N,Z}$ to the fermionic subspace $\wedge^{N} L^{2}(\Omega_{a})$ (the symbol $\wedge$ stands for the antisymmetric tensor product). To do so we introduce a projection $P^{AS}$ on $L^{2}(\Omega_{a})^{\otimes^{N}}$ as follows,
$$(P^{AS}\psi)(\bm{r}_{1},\ldots,\bm{r}_{N})=\frac{1}{N!}\sum_{\sigma\in S_{N}}\sgn{\sigma}~\psi(\bm{r}_{\sigma(1)},\ldots,\bm{r}_{\sigma(N)}).$$
Remark that this projection commutes with $H^{a}_{N,Z}$, i.e. $P^{AS}H^{a}_{N,Z}\subset H^{a}_{N,Z}P^{AS}$. On $\mathrm{Dom} (H^{a}_{N,Z})$, we define the {\it fermionized version} of the Hamiltonian $H^{a}_{N,Z}$ by
$$H^{a}_{N,Z,\mathrm{f}}:=P^{AS}H^{a}_{N,Z}P^{AS}=H^{a}_{N,Z}P^{AS}.$$ 
It is convenient to view $H^{a}_{N,Z,\mathrm{f}}$ as a restriction of $H^{a}_{N,Z}$ to $P^{AS}\mathrm{Dom}(H^{a}_{N,Z})$ acting in 
$$P^{AS}L^{2}(\Omega_{a})^{\otimes^{N}}\equiv\wedge^{N}L^{2}(\Omega_{a}).$$
Then $H^{a}_{N,Z,\mathrm{f}}$ is self-adjoint due to the following simple observation.
\begin{lemma}\label{lem:sa_reduced}
 Let $H$ be a self-adjoint operator on a Hilbert space $\mathscr{H}$ and $P$ be an orthogonal projection on $\mathscr{H}$. If
\begin{enumerate}
 \item $P\mathrm{Dom}(H)$ is dense in $P\mathscr{H}$
 \item $PH\subset HP$,
\end{enumerate}
then $H_{P}:=H|_{P\mathrm{Dom}(H)}$ is self-adjoint on $P\mathscr{H}$.
\end{lemma}
\begin{proof}
From the first condition and the self-adjointness of $H$, it follows $H_{P}\subset H_{P}^{\dagger}$. Next we have
$\mathrm{Ran}(H_{P}\pm i)=P\mathrm{Ran}(H\pm i)$ because of the second condition. But $\mathrm{Ran}(H\pm i)=\mathscr{H}$ by the self-adjointness criterion \cite[Thm. VIII.3]{rs1}. Using this criterion again we arrive at the assertion of the lemma.
\end{proof}

Similarly we define the fermionized versions of $H^{a}_{\mathrm{eff}}$ and $h_{N,Z}$ on  $\mathrm{Dom}(H_{\mathrm{eff}}^{a})$ and $\mathrm{Dom}(h_{N,Z})$, respectively,
\begin{align*}
 &H^{a}_{\mathrm{eff,f}}=P^{AS}P^{a}H_{N,Z}^{a}P^{a}P^{AS}=P^{a}H_{N,Z}^{a}P^{a}P^{AS}\\
 &h_{N,Z,\mathrm{f}}=P^{AS}_{\mathrm{2D}}h_{N,Z}P^{AS}_{\mathrm{2D}}=h_{N,Z}P^{AS}_{\mathrm{2D}},
\end{align*}
where $P^{AS}_{\mathrm{2D}}$ acts in the same manner as $P^{AS}$ but on the Hilbert space $L^{2}(\R^{2})^{\otimes^{N}}$. $H^{a}_{\mathrm{eff,f}}$ may be viewed as acting in $\wedge^{N}L^{2}(\R^{2})$ with domain $P^{AS}_{\mathrm{2D}}\mathrm{Dom}(H^{a}_{\mathrm{eff}})$ and the same action as $H^{a}_{\mathrm{eff}}$. An analogical statement holds true for $h_{N,Z,\mathrm{f}}$. Self-adjointness of $H^{a}_{\mathrm{eff,f}}$ and $h_{N,Z,\mathrm{f}}$ then again follows from Lemma \ref{lem:sa_reduced}. Moreover, $H^{a}_{\mathrm{eff,f}}$ and $h_{N,Z,\mathrm{f}}$ are bounded below with bounds greater than or equal to the lower bounds of $H^{a}_{\mathrm{eff}}$ and $h_{N ,Z}$, respectively. 

Going carefully through the proofs of Theorems \ref{theo:eff_2d_coul}, \ref{theo:total_eff_comp}, and \ref{theo:main} (and the related lemmas) one concludes that these theorems remain valid for the fermionized versions of operators too, if we everywhere interchange $d_{\mathrm{eff}}(\xi)$ and $d_{N,Z}(\xi)$ for $d_{\mathrm{eff,f}}(\xi)$ and $d_{N,Z,\mathrm{f}}(\xi)$, respectively. Here we introduced
$$d_{\mathrm{eff,f}}(\xi):=\mathrm{dist}(\xi,\sigma(H^{a}_{\mathrm{eff,f}})),\quad d_{N,Z,\mathrm{f}}(\xi):=\mathrm{dist}(\xi,\sigma(h_{N,Z,\mathrm{f}})).$$

\section*{Acknowledgments}
This work has been supported by the grant No. 13-11058S of the Czech Science Foundation (GA\v{C}R) and the project RC120002 of the Iniciativa Cientifica Milenio (ICM CHILE). The author wishes to express his thanks to Rafael D. Benguria for many fruitful discussions and to Pavel \v{S}\v{t}ov\'{i}\v{c}ek for stimulating suggestions during the preparation of the paper. The author also acknowledges useful remarks of Mark Ashbaugh on topics related to the Feshbach decomposition.

\end{document}